\documentclass[12pt]{article}
\usepackage{graphicx} 
 \usepackage[utf8]{inputenc}
\usepackage[ margin=2.5cm]{geometry}
\usepackage{dsfont}
\usepackage{setspace}
\usepackage{indentfirst}
\usepackage{xcolor}
\usepackage{times}
\usepackage{fancyhdr}
\usepackage{url}
\usepackage{float}
\usepackage{authblk}
\usepackage{subcaption} 
\usepackage{enumerate}
\usepackage{amsmath, amsthm, amsfonts, amssymb, amsxtra}
\usepackage{multirow}
\usepackage{adjustbox}
\usepackage{booktabs}
\usepackage{soul}
\usepackage{enumerate}
\usepackage{geometry}
\usepackage{hyperref}
\usepackage{lipsum}
\usepackage{tabularx}	
\usepackage[numbers]{natbib}

\newtheorem{theorem}{Theorem}
\newtheorem{lemma}{Lemma}

\theoremstyle{definition}

\title{On the formulation and analysis of a stage-structured eco-evolutionary model with pulsed disturbances}
\date{August 2026}
\author{Abigail D'Ovidio Long

Department of Mathematics, Computer Science and Statistics

Muhlenberg College

Allentown, PA 18104

Email: abbeydovidio@muhlenberg.edu\\
\and
Abigail Adjei

Department of Mathematics

Oregon State University 

Corvallis, OR 97331 

Email: adjeia@oregonstate.edu \\

\and
Swati Patel

Department of Mathematics

Oregon State University 

Corvallis, OR 97331 

Email: patelswa@oregonstate.edu \\}

\begin{document}

\maketitle
\newpage
\begin{abstract}
Motivated by the pulsed control and resulting evolution of harmful biological populations, we develop and analyze an ordinary differential equation model of an $n-$dimensional stage-structured population undergoing periodic disturbances. The eco-evolutionary model couples both ecological dynamics through life cycle information, and evolutionary dynamics of allele frequency changes. Building off prior work in eco-evolutionary models and theory of linear systems, we show that if a resistant allele is present in the population, the population will always approach full resistance. Furthermore, we establish persistence conditions and show that these depend solely on system dynamics at full resistance.  Finally, we demonstrate the utility of this model via an application to the Spotted-Winged Drosophila, an invasive species often treated periodically with insecticides.  We use a parameterized model to further understand the complexities of relationships between ecological and evolutionary features of a species and its resistance evolution due to pulsed disturbances. 
\end{abstract}
\textit{Keywords}: eco-evolutionary dynamics, impulsive differential equations, stage-structured population models, resistance evolution
\section{Introduction}

Many biological populations pass through distinct life stages and are exposed to disturbances that occur repeatedly over time. These disturbances may affect the different life stages differently as individuals transition through the population. Moreover, survival through the various stages may depend on genotype. Hence, repeated exposure to a disturbance can change both the abundance and the genetic composition of the population. This creates an interaction between the ecological dynamics of the population and the evolution of resistance. Stage-structured populations undergoing repeated treatment provide clear examples of coupled ecological and evolutionary dynamics, in cancer \cite{ dovidio2026, foo2014evolution, lakmeche_chemo_periodic_2000, ren2017tumour}, agricultural systems \cite{mailleret_pulsed_biocontrol_2009, mermer2021timing, tang_periodic_lv_2002}, and infectious diseases \cite{ gao2017mass, patel_spectral_2024, patel2025anthelmintic, WHO2006preventive}. For example, in the Spotted-Wing Drosophila, where insecticides can differentially affect egg, larval, pupal, and adult stages \cite{Mermer2021, Wiman2016}, periodic treatments alter population abundance while selecting for resistance alleles \cite{Ganjisaffar2022, GressZalom2019}. These systems illustrate how stage structure and pulsed control can jointly shape population persistence and the evolution of resistance. In this work, we couple classical population genetics equations with stage-structured population dynamic equations to formulate a general eco-evolutionary model, in which a periodic pulse disturbance potentially impacts each stage and genotype differently.

In the last few decades, there has been a substantial increase in the number of models formulated to examine the dual role of ecological and evolutionary processes on population and community dynamics. These are roughly divided into three different overarching frameworks, distinguished by their evolutionary assumptions and set up: 1) adaptive dynamics, (2) quantitative genetics, and (3) explicit population genetics. Adaptive dynamics is a modeling framework that assumes mutations are rare and often that selection is frequency-dependent and that this determines the long-term evolution of traits in ecological populations \cite{dieckmann1996dynamical, geritz1998evolutionarily, metz1996adaptive}. Quantitative genetics, on the other hand, assumes that traits are governed by infinitely many genetic loci, each with infinitesimally small effects, and hence, is often associated with a normal trait distribution assumption \cite{lande1976natural, pastore2021evolution, patel2015evolutionarily, schreiber2011community}.  Explicit population genetics models distinct loci (usually one \cite{schreiber2018evolution} or two \cite{patel2019ecoevolutionary}) and the frequencies of different alleles.  Indeed, numerous models of eco-evolutionary dynamics have highlighted how the coupled system can have qualitatively different dynamics, for example in stability \cite{cortez2016magnitude, patel2018partitioning} or coexistence \cite{patel2026persistence, patel2018robust}, compared to systems studied in an uncoupled manner. Hence, in the context of population control, these are critical to account for and examine the question: when does evolution fundamentally change our expected outcomes?

From an evolutionary perspective, there is a foundation of stage-structured models in the literature. This foundation includes both continuous and discrete models. An early contribution can be found in \cite{Taylor1990}, where fitness of mutant alleles are studied in a generic class-structure setting through continuous dynamics. In this work, they consider classes such as life stage, sex, and age. More recently, \cite{Barfield2011} proposed a general continuous framework for evolution in stage-structured populations. This model was intended for studying quantitative genetics, but is adaptable for allele and genotype frequencies. In \cite{deVries2019}, discrete matrix models were utilized to study population genetics in a species with three genotypes governed by two alleles. However, none of these works incorporate both continuous and discrete dynamics into a single model to study the implication of various time scales on the evolutionary process.

Combining continuous and discrete dynamics through pulsed ordinary differential equations can be found throughout the literature in modeling population dynamics. For example, \cite{Dong2006, He2015, Jiao2008, Jiao2007,  Jin2005, Kalra2022} and \cite{Li2022} study models of ecological problems through pulsed ODEs, whereas \cite{dovidio2026} uses the framework to analyze dynamics of cancer cells undergoing radiation therapy. In general, pulsed ODEs are a useful tool for systems in which dynamics are primarily continuous, except for at periodic times in which an abrupt change occurs, (relatively quickly compared to life cycle changes), such as in harvesting or stocking. Most pertinent to our work is \cite{patel2024spectral}, which considers the control and resistance of a stage-structured population through pulsed massive drug administration. Notably, this previous work only considered two stages in which only one stage is impacted by control. In the present work, this is generalized to arbitrary many stages in which each stage can be impacted differentially and hence, vulnerable to selection in all stages due to the pulse. 

The goals of this work are three-fold. In Section \ref{model_formulation}, we formulate a stage-structured pulsed ODE model which encapsulates both ecological and evolutionary dynamics for a population with $n$ stages undergoing periodic control. The model is a generalized version of that presented in \cite{patel2026persistence}, where a two-dimensional system is presented in which only population members in the adult stage are impacted by the pulse. In Section \ref{analysis}, we give thorough analysis of asymptotic behavior of the model by considering isolated cases before proving the main results on ecological and evolutionary dynamics for the entire system. In Section \ref{numerical}, we apply our model to the Spotted-Winged Drosophila, as a relevant example of a pest in an agricultural setting. 

\section{Model Formulation}\label{model_formulation}
Here, we consider a single population structured into $n$ life stages and let $x(t)=(x_1(t),...,x_n(t))$ be the population densities of each stage at time $t$.  We assume that stage $i$ has a per-capita mortality rate of $\mu_i$, developmental transitions between stage $i$ and $j=i+1$ have rate $\alpha_{ij}$, and stage $n$ is a fully mature stage that produces new offspring that enter into stage $1$ at rate $\lambda$.  These rates are assumed to be density-independent. 

Additionally, we assume that the population is disturbed, periodically every $\tau$ time units, and this instantaneously causes a mortality to some proportion of the population.  We assume that the impact of the disturbance is both stage-dependent and genotype-dependent.  That is, the probability an individual survives the disturbance depends on which stage it is in and its genetics.  We model the genotype-dependence by assuming there is a single locus with two alleles, $1$ and $2$.  We say that the probability that an individual in stage $i$ of genotype $kj$ survives the disturbance is given by $\delta_i^{kj}$. Throughout, we assume that, for each stage $i$,
\begin{equation}
    \delta^{11}_i\leq \delta^{12}_i\leq \delta^{22}_i
    \label{eq:delta}
\end{equation} 
which represents the case that allele $2$ is a resistant allele. 
Under this assumption, we define the survival based on the selection coefficient $s$, heterozygosity $h$, following standard formulation from population genetics:

\begin{equation}
    \delta^{11}_i = (1-s)\delta^{22}_i \quad \delta^{12}_i = (1-sh)\delta^{22}_i
    \label{eq:HW}
\end{equation}

The selection coefficient captures the relative survival of the susceptible to resistant genotypes and is assumed to be between 0 and 1.  The heterozygosity captures the dominance of the resistant or susceptible alleles and is also assumed to be between 0 and 1. (If $h$ is close to 0, then resistance is dominant, while if $h$ is close to 1, then susceptibility is dominant.)

We make some simplifying assumptions.  First, we assume that no other life history traits (e.g. fecundity, natural mortality) are dependent on genotype.  That is, there is no fitness cost to resistance.  This is appropriate in some cases, where there is little to no evidence of costs to resistance, such as in some helmintic parasites \cite{ borgsteede1989lack, dilks2021newly}.  It also serves as the purpose of being a baseline model from which future extensions can be made and compared to these results. Then, the population dynamics outside of the disturbance are genotype-independent and given by

\begin{equation}
    \frac{dx}{dt} = A_{\text{eco}}x
    \label{eq:ecodynamics}
\end{equation}
where $A_{\text{eco}}$ is a matrix that captures the mortality, fecundity and developmental rates.  An example for when $n=3$ is
\begin{equation}
    A_{\text{eco}}=\begin{bmatrix}
    -(\mu_1+\alpha_{12})  & 0 & \lambda \\
     \alpha_{12} & -(\mu_2+\alpha_{23}) & 0\\
     0 &\alpha_{23} & -\mu_3 \\
    \end{bmatrix};
    \label{eq:A_eco}
\end{equation}

We let $y_i(t)$ be the frequency of allele $2$ in stage $i$.  A second simplifying assumption: we assume Hardy-Weinberg equilibrium so that the frequency of allele $2$ determines the genotype frequencies in each stage. By tracking the frequency changes caused by developmental transitions between stages, we derive the dynamics of the frequency of allele $2$ as  
\begin{equation}
    \frac{dy}{dt} = A_{\text{evo}}(x)y
    \label{eq:evodynamics}
\end{equation}
where now $A_{\text{evo}}$ depends on the population densities $x$ in each stage.  An example again in 3 stages is 
\begin{equation}
    A_{\text{evo}}=\begin{bmatrix}
    -\lambda \frac{x_3}{x_1}  & 0 &\lambda \frac{x_3}{x_1} \\
      \alpha_{12}\frac{x_1}{x_2} & -\alpha_{12}\frac{x_1}{x_2} & 0\\
     0 &\alpha_{23}\frac{x_2}{x_3} & -\alpha_{23}\frac{x_2}{x_3}
    \end{bmatrix}
    \label{eq:A_evo}
\end{equation}
A full derivation of this is given in the Supplementary Material. In line with intuition, the rates of changes to the allele frequency of stage $i$ is inversely proportional to the population density of stage $i$ and proportional to the population density of the stage feeding in to stage $i$. 

Then, to model the impact of the pulse disturbance, we define the mean survival of the population in stage $i$ in response to the pulse as it depends on the frequency of allele $2$:
\begin{align}\label{eq:mean_survival}
\bar\delta_i = y_i^2\delta^{22}_i + 2y_i(1-y_i)\delta^{12}_i + (1-y_i)^2\delta^{11}_i
\end{align}
We also define the marginal survival probability (of individuals with at least one resistant allele) in stage $i$:
\begin{align}\label{eq:marginal_survival}
\bar\delta_i^m = y_i\delta^{22}_i + (1-y_i)\delta^{12}_i
\end{align}

From this, we derive the impact of the pulse disturbance on the population densities and the frequency of resistance allele in each stage.
Our full model equations are formulated as the impulse differential equations

\begin{align} \label{eq:pulse}
\text{continuous} & \begin{cases}
\frac{dx}{dt} &= A_{\text{eco}}x \\
\frac{dy}{dt} &= A_{\text{evo}}(x) y \\
\end{cases}\\
\text{pulse} & \begin{cases}
x(k\tau^+) &= D_{\text{eco}}(y)x(k\tau^-) \\
y(k\tau^+) &= D_{\text{evo}}(y)y(k\tau^-)
\end{cases}
\end{align}
where $D_{\text{eco}}(y)$ is a diagonal matrix in which the $i^{th}$ diagonal element is the mean survival $\bar{\delta}_i$ (which depends on $y_i$). Similarly, $D_{\text{evo}}(y)$ is a diagonal matrix in which the $i^{th}$ diagonal element the ratio of the marginal survival to the mean survival: $\bar\delta_i^m/\bar\delta_i$ (as in classical population genetics). 
Note these equations have been derived from careful tracking of densities of genotypes. More details of the derivation are given in the Supplementary Material. 

The state space for these dynamics is $\mathbb{R}^n_+ \times [0,1]^n$, where $\mathbb{R}^n_+ = \{x | x_i>0\}$.  This forms a forward invariant set. Our goal here is to analyze the properties and asymptotic behavior of this coupled eco-evolutionary model equations, focusing on the persistence and evolution of resistance due to the pulsed control. We observe that in this model the ecological dynamics are decoupled from the genetics, in the absence of the pulse disturbance.  Selection only occurs instantanesously at the pulsed moments and these are decoupled from the ecological variables. Hence, eco-evolutionary feedbacks emerge from continuous evolutionary dynamics and the pulse impacts on the ecology.  

We begin with a preliminary result that shows that at the moment of the pulse, frequency of the resistant allele cannot decrease, assuming that some proportion of the population in that stage survives. 
\begin{lemma}\label{result_nondec_pulse}
Assume that for all stages $i = 1, 2, \dots, n$, mean survival rate $\bar{\delta}_i>0$. Then 
\begin{equation}
y_i(k\tau^+)\ge y_i(k\tau^-).
\label{eq:pulse_nonnegative_increase}
\end{equation}
That is, the pulse disturbance has a non-decreasing impact on resistant allele frequency.
\end{lemma}

\begin{proof}
Set 
\[
y_i^-=y_i(k\tau^-),
\qquad
y_i^+=y_i(k\tau^+).
\]
By the formulation of $D_{evo}(y)$, 
\begin{align}
y_i^+-y_i^-
=
\frac{\bar{\delta}_i^{\,m}}{\bar{\delta}_i}y_i^- - y_i^-
=y_i^-\left(\frac{\bar{\delta}_i^{\,m}-\bar{\delta}_i}{\bar{\delta}_i}\right).
\end{align}
As \(0\le y_i^-\le 1\) and $\bar{\delta}_i>0$, it is will suffice to show that \(\bar{\delta}_i^{\,m}-\bar{\delta}_i \geq 0\).
From (\ref{eq:mean_survival}) and (\ref{eq:marginal_survival}), 
\begin{align}\label{eq:pulse_difference_formula}
\notag \bar{\delta}_i^{\,m}-\bar{\delta}_i
&=
\bigl(y_i^-\delta_i^{22}+(1-y_i^-)\delta_i^{12}\bigr)
-
\bigl((y_i^-)^2\delta_i^{22}+2y_i^-(1-y_i^-)\delta_i^{12}+(1-y_i^-)^2\delta_i^{11}\bigr) \\
&= (1-y_i^-)
\Bigl[
y_i^-(\delta_i^{22}-\delta_i^{12})
+
(1-y_i^-)(\delta_i^{12}-\delta_i^{11})
\Bigr].
\end{align}
 By \ref{eq:delta},
\begin{align}
\delta_i^{22}-\delta_i^{12}\ge 0,
\qquad
\delta_i^{12}-\delta_i^{11}\ge 0.
\end{align}
Therefore, every factor on the right-hand side of \eqref{eq:pulse_difference_formula} is nonnegative, and so
$\bar{\delta}_i^{\,m}-\bar{\delta}_i \geq 0.$
Thus, $y_i(k\tau^+) \geq y_i(k\tau^-)$ for all $i = 1, 2, \dots, n$.
\end{proof}
We will assume throughout, then, that $\bar{\delta}_i>0$.
\section{Analysis}\label{analysis}
Our model is a coupled ecological and evolutionary model with a pulse disturbance that has a selective and demographic effect.  
We partition our analysis by presenting results on different components of this model and then for the coupled system. More specifically, we first review existing results on the pure ecological components. Then, we present results for the pure evolutionary components. Finally, we present an analysis of the joint coupled system.

\subsection{Pure Ecology}
In this section, we review some well-known results as they apply to our ecological model (assuming evolutionary states are fixed).  We include them for the sake of completion and since we rely on them for analysis of the coupled systems. 

\subsubsection{Without Pulse}
Given any matrix $A$, it is well known that solutions of differential equations of the form

\begin{equation}\label{eco_nopulse}
\frac{dx}{dt} = Ax
\end{equation}
are $x(t) = e^{At}x(0)$.  Additionally, defining the spectral abscissa of $A$, $s(A)$, as the largest real part of the eigenvalues, then we can determine asymptotic behaviors of solutions.  Mainly, that if $s(A)>0$, then $x(t)$ diverges provided $x(0)\neq 0$.  Conversely, if $s(A)<0$, then $x(t)\rightarrow 0$, i.e., the population goes extinct.  

Furthermore, if, in addition, $A$ is essentially non-negative (non-negative off the diagonal) and irreducible, then, by an application of the celebrated Perron-Frobenius theorem, we get that $s(A)$ is an eigenvalue of $A$ and that the associated eigenvector $v = (v_1, v_2, ... v_n)^T$ is positive and the only positive eigenvector. From this, we have that for any solution $x(t)$ of \ref{eco_nopulse} with an initial condition $x(0)>0$, the following is true

\begin{equation}
    \lim_{t\rightarrow \infty} \frac{x_i(t)}{x_j(t)} = \frac{v_i}{v_j}
\end{equation}
for any stages $i$ and $j$. We use this later in the analysis of the coupled eco-evo pulsed system.

\subsubsection{With Pulse}
We next review some recent results on linear impulse systems of the form

\begin{align}
\frac{dx}{dt} &= Ax\\
x(k\tau^+) &= Dx(k\tau^-) 
\end{align}
where $D$ is a diagonal matrix. 

In this case, it was shown in \cite{patel2024spectral} that the asymptotic behavior of the impulsed system is determined by the spectral radius of matrix $De^{\tau A}$, denoted $r(De^{\tau A})$.  If $r(De^{\tau A})>1$, then $x(t)$ diverges, while if $r(De^{\tau A})<1$, then $x(t)\rightarrow 0$. This distinction is critical in determining the effects of resistance evolution. As the frequency of resistance alleles increase, entries in the diagonal matrix $D$ increase and hence, may lead to qualitatively different asymptotic behaviors. 

We numerically examined the relationship between the life cycle transitions (parameters in matrix $A$) and the stage-dependent impact of the disturbance (parameters in $D$) for a three-stage population (Figure~\ref{fig:spectral_radius}). We considered four parameter combinations to examine the effects of the transition rates $\alpha_{12}$ and $\alpha_{23}$ and the fecundity rate $\lambda$. Panels (A) and (B) use $\lambda=0.35$, while panels (C) and (D)
use $\lambda=0.40$. Within each pair, left side figures use $\alpha_{12} = 0.2$ and $\alpha_{23} = 0.45$, with the values flipped for right side figures. This allows us to compare the impact of fast and slow transition rates between stage, along with the impact of varying fecundity rates. Persistence boundaries at $r\left(De^{\tau A}\right)$ are plotted for various fixed levels of $\delta_3$ across axes of $\delta_1$ and $\delta_2$. Points above the boundary will result in persistence, while points below will result in extinction.

\begin{figure}[h]
    \centering
    \includegraphics[width=\textwidth]{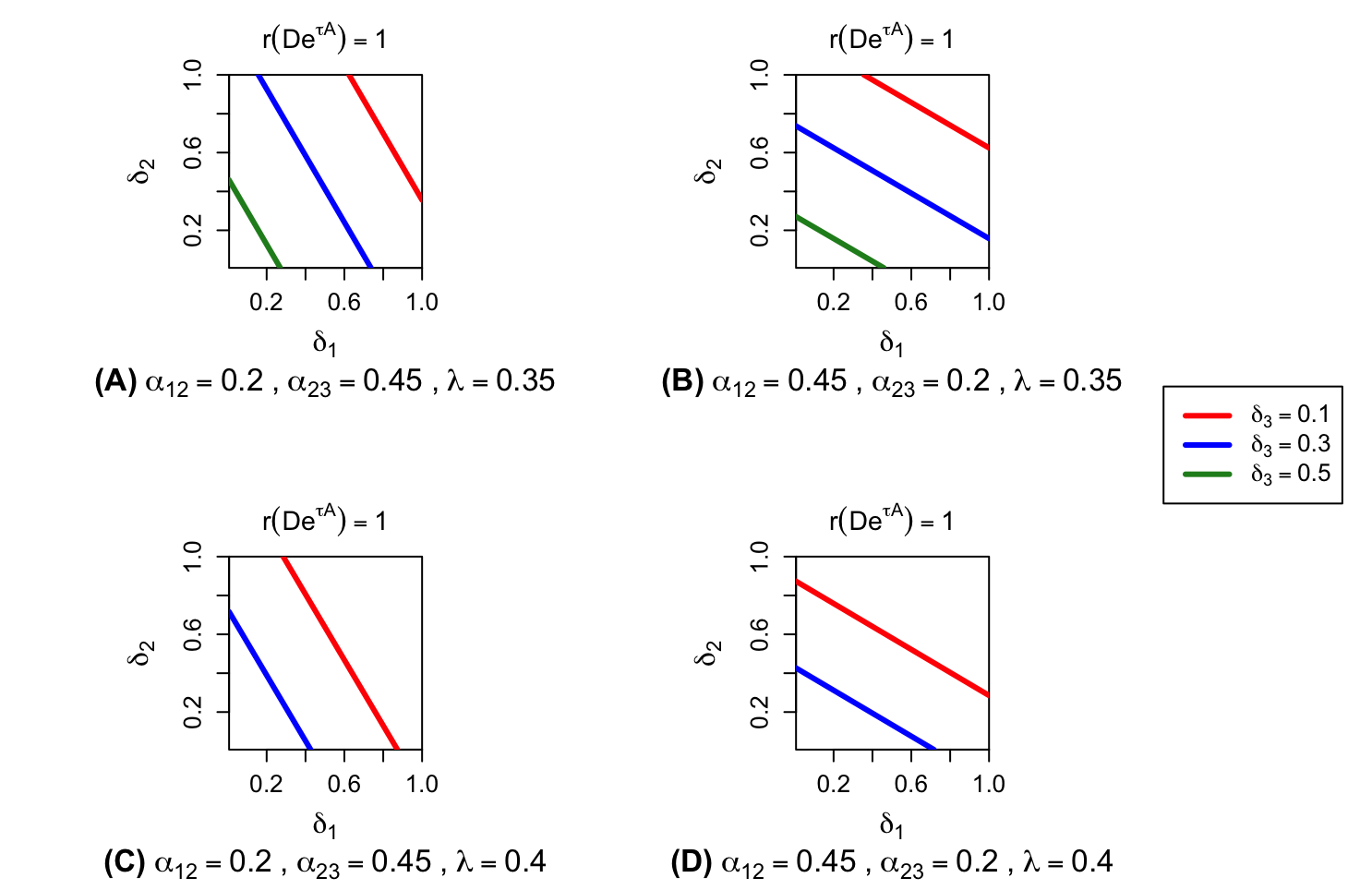}
    \caption{Persistence boundaries satisfying
    $r(De^{\tau A})=1$ in the $(\delta_1,\delta_2)$ plane for three fixed
    values of the Stage 3 survival probability, $\delta_3\in\{0.1,0.3,0.5\}$. The horizontal axis represents the Stage 1 survival probability ($\delta_1$), and the vertical axis represents the Stage 2 survival probability ($\delta_2$). Each contour represents the persistence threshold for the specified value of $\delta_3$, with points above the curve corresponding to $r(De^{\tau A})>1$ and persistence, and points below the curve corresponding to $r(De^{\tau A})<1$ and extinction. Panels (A) and (B) use $\lambda=0.35$, while panels (C) and (D) use $\lambda=0.40$; within each pair, the transition rates $\alpha_{12}$ and $\alpha_{23}$ are interchanged.}
    \label{fig:spectral_radius}
\end{figure}

Across all four panels of Figure \ref{fig:spectral_radius}, the persistence boundaries have negative slopes, indicating a compensatory relationship between survival in Stages 1 and 2. Along the threshold $r(De^{\tau A})=1$, an increase in survival in one of these stages can compensate for a decrease in survival in the other. Additionally, greater
survival in Stage 3 shifts the persistence boundary downward, allowing the population to persist even at lower combinations of $\delta_1$ and $\delta_2$. Conversely, when Stage 3 survival is low, greater survival in the first two stages is required for persistence.

In panel (A), where $0.20 = \alpha_{12}<\alpha_{23}=0.45$, the contours are relatively shallow, indicating that the persistence threshold is more sensitive to changes in Stage 2 survival than to changes in Stage 1 survival. Individuals move more slowly from Stage 1 into Stage 2 but more rapidly from Stage 2 into the reproductive stage. Consequently, survivors in Stage 2 contribute more directly to recruitment into Stage 3, making changes in $\delta_2$ particularly influential for persistence. The opposite pattern occurs in panel (B), where $0.45=\alpha_{12}>\alpha_{23}=0.20$. Here, the contours are steeper, indicating greater sensitivity of the persistence threshold to Stage 1 survival. Individuals move relatively quickly from Stage 1 into Stage 2 but more slowly from Stage 2 into the reproductive stage. Under this transition structure, reductions in $\delta_1$ have a stronger effect on the persistence threshold than comparable reductions in $\delta_2$. The same pattern appears when comparing panels (C) and (D).

Comparisons of panels (A) and (C), and similarly panels (B) and (D), isolate the effect of increasing fecundity from $\lambda=0.35$ to $\lambda=0.40$. Increasing $\lambda$ shifts the persistence boundaries downward, so that persistence can occur at lower combinations of Stage 1 and Stage 2 survival, and making extinction solutions more difficult to attain. In fact, at the $\delta_3 = 0.5$ or 50\% survival rate in stage three, it is no longer possible to achieve extinction. Thus, both the rates of progression through the life cycle and fecundity influence the combinations of stage-specific survival probabilities required for population persistence.

\subsection{Pure Evolution}
In this section, we provide results for the evolutionary dynamics.  That is, we assume that the population densities are all constant in time, i.e., $x(t)=x \in \mathbb{R}^n_+$.  

\subsubsection{Without Pulse}
In the continuous dynamics, we observe that the rows in the matrix $A_{\text{evo}}$ sum to $0$, so it can be viewed as an infinitesimal generator of an (irreducible) continuous-time Markov chain. From this, we observe that 
there are infinitely-many equilibrium solutions for $y$, along the line formed by $y_1=y_2=...=y_n$ in $[0,1]^n$.  

The following result shows that the asymptotic behavior of the frequencies of resistance in each stage is to approach one of these equilibrium points.  Which equilibrium depends on the initial frequencies and we provide an expression. 

\begin{theorem}\label{evo_1}
Solutions to the differential equation
\begin{equation}
    \frac{dy}{dt} = A_{\text{evo}}y
    \label{eq:evo}
\end{equation}
with initial condition $y(0)=(y_{o,1}, y_{o,2},..., y_{o,n})$ satisfy 
\begin{equation}
    \lim_{t\rightarrow \infty} y_i(t) = \frac{\sum_{j=1}^n\prod_{i\neq j} [A_{\text{evo}}]_{ii}y_{o,j}}{\sum_{j=1}^n\prod_{i\neq j} [A_{\text{evo}}]_{ii}}
    \label{eq:evo_limit}
\end{equation}
for all $i$.
\end{theorem}

\begin{proof}

To prove this result, we first define the conserved quantity:
\begin{equation}
    z = -\sum_{j=1}^n\prod_{i\neq j} [A_{\text{evo}}]_{ii}y_j
    \label{eq:conserved_quantity}
\end{equation}

This is straightforward to show, by writing down the time derivative of $z$ and showing it equals zero. These constants of motion form $(n-1)-$dimensional (hyper)planes, which are forward invariant subspaces and compact in $[0,1]^n$. 

Secondly, we note that 
\begin{equation}
    \lim_{t\rightarrow \infty} y_i(t) = \lim_{t\rightarrow \infty} y_{i-1}(t)
    \label{eq:yi_equal_limit}
\end{equation}
for all $i$ (and where subscripts are mod $n$).    
Notice that $A_{evo}$ is an essentially non-negative matrix (non-negative off the diagonal) with a right-most eigenvalue of zero and associated eigenvector of 1 (again from Perron-Frobenius). 
Then, the result follows.

\end{proof}

We show a graphical version of this for the case $n=2$ in Figure \ref{fig:constants_of_motion}.  
\begin{figure}[H]
    \centering

    \begin{subfigure}[t]{0.4\textwidth}
        \centering
        \includegraphics[width=\textwidth]{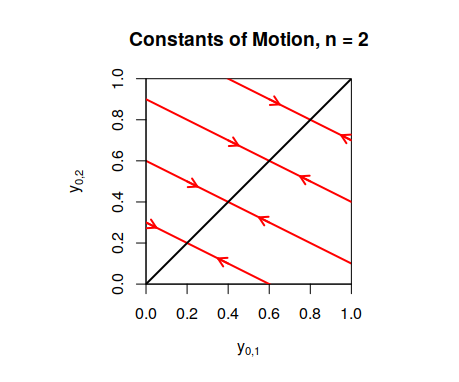}
        \caption{$\alpha_{12} = 0.2$, $\lambda = 0.4$}
        \label{fig:slow}
    \end{subfigure}
    \begin{subfigure}[t]{0.4\textwidth}
        \centering
        \includegraphics[width=\textwidth]{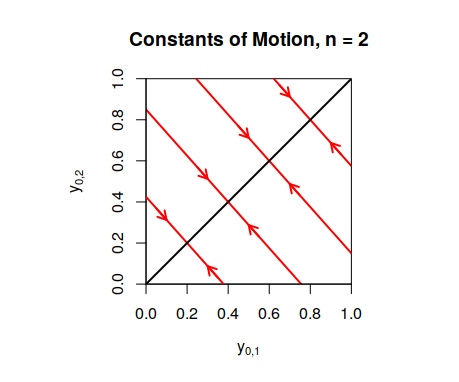}
        \caption{$\alpha_{12} = 0.45$, $\lambda = 0.4$}
        \label{fig:fast}
    \end{subfigure}
    \caption{A graphical representation of Result \ref{evo_1} in two dimensions for two cases. Subfigure (A) reflects the slow transition rate between the two stages in Figure \ref{fig:spectral_radius}, whereas Subfigure (B) reflects the fast transition rate.  The black line is $y_{0,1} = y_{0,2}$ and the red lines are the conserved quantity as in equation (\ref{eq:conserved_quantity}) for different fixed values of $z$.}
    \label{fig:constants_of_motion}
\end{figure}

Additionally, for generalizing, it helps to see the expression written out for the case when $n=3$. That is

\begin{align}
\lim_{t\rightarrow \infty} y_i(t)= \frac{(\alpha_{12}\alpha_{23}\frac{x_1}{x_3})y_{o,1} + ( \alpha_{23}\lambda\frac{x_2}{x_1})y_{o,2} + (\lambda\alpha_{12}\frac{x_3}{x_2}) y_{o,3}}{\alpha_{12}\alpha_{23}\frac{x_1}{x_3}+ \alpha_{23}\lambda\frac{x_2}{x_1}+ \lambda\alpha_{12}\frac{x_3}{x_2}}
\end{align}
where $\lambda$ and $\alpha_{ij}$ are as in equation \ref{eq:A_eco}.
What this shows is that allele frequencies will ``equalize" during the continuous evolutionary dynamics as individuals mix between stages. We observe also that the transition rates affect whether the frequency in the first stage has a bigger influence on the final frequency than the second stage or vice versa. (The more vertical the contours in Figure \ref{fig:constants_of_motion}, the less the frequency in stage 1 changes from the mixing.)

\subsubsection{With Pulse}

\begin{theorem}\label{evo_wpulse}
    Assume there is at least one $1 \leq j \leq n$ such that selection coefficient $s_j>0$. Then if fixed population density vector $x > 0$, solutions to the impulse differential equation 
\begin{align}
\frac{dy}{dt} &= A_{\text{evo}}y\\
    y(k\tau^+) &= D_{\text{evo}}(y)y(k\tau^-)
\end{align}
with initial condition $y(0)\neq 0$ satisfy 
\begin{equation}
    \lim_{t\rightarrow \infty} y_i(t) = 1
    \label{eq:evo_pulse_limit}
\end{equation}
    for all $1 \leq i \leq n$
\end{theorem}

\begin{proof}
Let $V(y) = \displaystyle\min_{1 \leq i \leq n}\{y_i\}$.  Then, $V$ is differentiable with respect to time for all values of $t \neq k\tau$ with the possible exception of at most countably many points  where $V$ switches from $V = y_i$ to $V = y_j$ for some $i \neq j$, as $y$ is smooth component-wise at $t \neq k\tau$. Moreover, we claim $\frac{dV}{dt}\geq 0$ where it exists. To prove this assertion, by way of contradiction, suppose that for some value of $t$, $V(y) = y_j$ and that $\frac{dV}{dt} = \frac{dy_j}{dt} < 0$. By the formulation of $A_{evo}$, (illustrated in three dimensions in equation (\ref{eq:A_evo})) 
\begin{align}
    \frac{dy_j}{dt} = k(y_i - y_j)
\end{align}
for a positive constant value $k$ and $y_i \neq y_j$. Hence, if $\frac{dy_j}{dt} < 0$, it must be true that $y_i < y_j$, and thus $V \neq y_j$. 

Furthermore, $V(k\tau^+)\geq V(k\tau^-)$, since the effects of the pulse disturbance on $y$ are always non-decreasing component-wise as shown in Result \ref{result_nondec_pulse}. Hence, since $s_j>0$ for at least one $j$, $V$ is a non-decreasing function in time, for which the only possible fixed points with respect to our pulsed model are at $V=0$ or $V=1$. However, if at least one $y_j(0) > 0$, then by the coupled structure of $A_{evo}$, for each $y_i$ there exists some $t$ such that $\frac{dy_i}{dt} > 0$. So, $V = 0$ is only a fixed point if $y(0) = 0$. Thus, provided $y(0) \neq 0$, $\displaystyle \lim_{t \to \infty} V  = 1$  which implies $\displaystyle\lim_{t \to \infty} y_i = 1$ for all $i$.
\end{proof}

\subsection{Eco-Evolutionary Dynamics}
In this section, we analyze the coupled ecological and evolutionary dynamics. 
\subsubsection{Without Pulse}
We first examine the dynamics of the eco-evolutionary model without the pulse
\begin{align}\label{eq:eco_evo}
\frac{dx}{dt} &= A_{\text{eco}}x \\
\frac{dy}{dt} &= A_{\text{evo}}(x) y \label{eq:eco_evo2} 
\end{align}\\
Notice that without the pulse, these are decoupled. That is, the ecological dynamics are not affected by the evolutionary dynamics.  (This is because of our assumption in this baseline model that the genotypes only impact the response to the pulsed disturbance and have no effect on the life history.) From our results in part 1 (pure ecology), the solution $x(t)$ to (\ref{eq:eco_evo}-\ref{eq:eco_evo2}) with initial conditions $x(0)>0$ approaches the solution $x^*(t) = ve^{\eta t}$ where $v$ and $\eta$ are the eigenpair guaranteed by the Perron-Frobenius theorem. Hence, we use the theory of asymptotically autonomous systems and examine 

\begin{align}\label{eq:asymptotic}
\frac{dy}{dt} &= A_{\text{evo}}(x^*) y
\end{align}
to determine limiting behaviors of solutions $y(t)$ to (\ref{eq:eco_evo}-\ref{eq:eco_evo2}). Before the result, we restate an important theorem.
\begin{theorem}[Markus \cite{Markus1956} as stated in Thieme \cite{Horst1994}]\label{thm_markus}
    For an asymptotically autonomous ordinarily differential equation 
    \begin{align}\label{eq:asymptotic_def}
        \dot{x} = f(t, x) 
    \end{align}
    with limiting equation 
    \begin{align}\label{eq:limiting}
        \dot{y} = g(y) 
    \end{align}
    then the $\omega-$limit set of a forward bounded solution $x$ of \eqref{eq:asymptotic_def} is non-empty, compact and connected. Moreover, $\omega$ attracts $x$, or 
    \begin{align}
        \text{dist}(x(t), \omega) \to 0
    \end{align}
    as $t$ approaches infinity. Finally, $\omega$ is invariant under \eqref{eq:limiting}. In particular, any point in $\omega$ lies on a full orbit of \eqref{eq:limiting} that is contained in $\omega$. 
\end{theorem}
With this theorem, we are able to show that all stages will approach a fully resistant population.
\begin{theorem}
The solution $y(t)$ to the system (\ref{eq:eco_evo}-\ref{eq:eco_evo2}) converges to $c[1, 1, \dots, 1]^T$ for some $c \in [0, 1]$. That is, $\displaystyle\lim_{t \to \infty} y_i(t) = \lim_{t \to \infty} y_j(t)$ for all $i, j \in 1, 2, \dots, n$.
\end{theorem}
\begin{proof}
 First, we will show that $A_{evo}(x)$ converges uniformly to $A_{evo}(x^*)$, an autonomous matrix. Set $K = \displaystyle\max_{1\leq i \leq n-1}\{\lambda, \alpha_{i, i+1}\}$ and $x^*(t) = [x_1^*, x_2^*, \dots, x_n^*]^T$. Then 
\begin{align*}
\left| A_{evo}(x)y - A_{evo}(x^*)y\right| &= \left|\begin{bmatrix}
    \lambda \frac{x_n}{x_1}(y_n - y_1) \\
    \alpha_{12} \frac{x_1}{x_2} (y_1 - y_2)\\
    \vdots\\
    \alpha_{n-1,n} \frac{x_{n-1}}{x_n} (y_{n-1} - y_{n})  \\
\end{bmatrix} - \begin{bmatrix}
    \lambda \frac{x_n^*}{x_1^*}(y_n - y_1) \\
    \alpha_{12} \frac{x_1^*}{x_2^*} (y_1 - y_2)\\
    \vdots\\
    \alpha_{n-1,n} \frac{x_{n-1}^*}{x_n^*} (y_{n-1} - y_{n})  \\
\end{bmatrix} \right| \\
&\leq K \left| \begin{bmatrix}
    (\frac{x_n}{x_1} - \frac{x_n^*}{x_1^*}) (y_n - y_1) \\
    (\frac{x_1}{x_2} -  \frac{x_1^*}{x_2^*}) (y_1 - y_2)\\
    \vdots \\
    (\frac{x_{n-1}}{x_n} -  \frac{x_{n-1}^*}{x_n^*}) (y_{n-1} - y_{n})\\
\end{bmatrix}\right|\\
&\leq K \left| \begin{bmatrix}
    (\frac{x_n}{x_1} - \frac{x_n^*}{x_1^*})  \\
    (\frac{x_1}{x_2} -  \frac{x_1^*}{x_2^*})\\
    \vdots \\
    (\frac{x_{n-1}}{x_n} -  \frac{x_{n-1}^*}{x_n^*})\\
\end{bmatrix}\right|.
\end{align*}
Thus, showing that $A_{evo}(x)$ converges uniformly to $A_{evo}(x^*)$ is equivalent to showing that for all $\varepsilon > 0$, there exists $T > 0$ such that 
\begin{align}\label{eq:ep_result}
\left| \frac{x_i}{x_j} - \frac{x_i^*}{x_j^*}\right| < \varepsilon
\end{align}
for all $t > T$ and $i = j -1$ for $j = 2, 3, \dots, n$ or $i = n$ for $j = 1$. Note that as $x$ converges to $x^*$, it follows that 
\begin{align*}
    \lim_{t \to \infty} \frac{x_i}{x_i^*} = 1, \hspace{2mm} \lim_{t \to \infty} \frac{x_j^*}{x_j} = 1
\end{align*}
for all $i, j = 1, 2, \dots, n$. Thus, 
\begin{align*}
     \lim_{t \to \infty} \frac{x_i/x_j}{x_i^*/x_j^*} = \lim_{t \to \infty} \frac{x_i}{x_i^*}\left(\frac{x_j^*}{x_j}\right) = 1.
\end{align*}
Hence, $\displaystyle\lim_{t \to \infty} \dfrac{x_i}{x_j} =\dfrac{x_i^*}{x_j^*} $, and so there exists a $T_i > 0$ such that (\ref{eq:ep_result}) holds for $t > T_i$ for each $i, j$ pair. Choosing $T =\displaystyle \max_{i = 1, 2, \dots, n} T_i$ yields the necessary result to show that $A_{evo}(x)$ converges uniformly to $A_{evo}(x^*)$, and thus (\ref{eq:eco_evo}) is asymptotically autonomous. \\
Next, we consider the subsets of $[0, 1]^n$ which are invariant under the flow of (\ref{eq:asymptotic}). Note that the set $\mathcal{C} := \{y = c[1, 1, \dots, 1]^T : c \in[0, 1]\}$ forms an invariant set, as for each such $y$, $A_{evo}(x^*)y = \vec{0}$, and thus each is a fixed point. As shown in Result \ref{evo_1} in Section 3.2, solutions to the autonomous system converge to this limit. Hence, the only invariant subset of $[0,1]^n$ under system (\ref{eq:asymptotic}) is $\mathcal{C}$. 
As the $\omega-$limit sets of (\ref{eq:eco_evo}) are invariant under (\ref{eq:asymptotic}), they must be contained in $\mathcal{C}$ by Theorem \ref{thm_markus}. Furthermore, since any bounded solution $y(t)$ of (\ref{eq:eco_evo}) converges to the $\omega-$limit set, it follows that such solutions will converge to a vector of the form $c[1, 1, \dots, 1]^T$ for some $c \in [0,1]$. 
\end{proof}

\subsubsection{With Pulse}
\begin{theorem}\label{result_4}
    Assume there is at least one $1 \leq j \leq n$ such that selection coefficient $s_j>0$. Then if  $x(0) > 0$ and $y(0) \neq 0$,  solutions to the full model defined in (\ref{eq:pulse}) satisfy  $\displaystyle\lim_{t \to \infty} y(t)= 1$.
\end{theorem}

\begin{proof}
Following the same proof as in Result \ref{evo_wpulse}, set $V(y) = \displaystyle\min_{1 \leq i \leq n}\{y_i\}$. As $\frac{dx}{dt}$ is a linear system for $t \neq k \tau$, it follows that $x_j \neq 0$ for all $j$ and $t \geq 0$ if $x(0) > 0$. Thus, $V$ is differentiable with respect to time for all values of $t \neq k\tau$ with the possible exception of at most countably many points where $V$ switches from $V = y_i$ to $V = y_j$ for some $i \neq j$. Also, $\frac{dV}{dt} \geq 0$ at points where the derivative exists by the same argument presented in Result \ref{evo_wpulse}, as if $t$ is fixed, $\frac{x_i(t)}{x_j(t)}$ is a positive constant value. The rest of the proof follows exactly from Result \ref{evo_wpulse}, as the pulse equation is identical and $\frac{dy}{dt}$ is still fully coupled. Hence, provided $y(0) \neq 0$, $\displaystyle \lim_{t \to \infty} V  = 1$  and so $\displaystyle\lim_{t \to \infty} y_i = 1$ for all $i$.
\end{proof}

\begin{theorem}\label{result_5}
    If $r(\tilde{D}e^{\tau A}) < 1$, where $\tilde{D}$ is the diagonal matrix $D_{eco}(1)$ and $A = A_{\text{eco}}$, then for any initial condition with $x(0)>0$ and $y_j(0) > 0$ for some $j$, we have $\displaystyle \lim_{t \to \infty} x(t)= 0$. 
\end{theorem}

\begin{proof}
Suppose that $r(\tilde{D}e^{\tau A}) < 1$ and $x(0) > 0$. Consider the discrete time model 
\begin{align}\label{eq:discrete_model}
x_{k+1} = D_{\text{eco}}(y(k\tau))e^{\tau A}x_k.
\end{align} 
 We claim that $\displaystyle\lim_{k \to \infty} x_k = 0$. To show this, 
note that for all $y(k\tau)$, $D_{\text{eco}}(y(k\tau)) \leq \tilde{D}$, where inequality is entry-wise. That is, $D_{\text{eco}}(y(k\tau)_{ii} \leq \tilde{D}_{ii}$ for all $i$ (since $D_{\text{eco}}(y(k\tau))_{ij} = \tilde{D}_{ij} = 0$ for all $i \neq j$). Also, as the off-diagonal entries of $\tilde{D}e^{\tau A}$ are non-negative, for any two positive vectors $x_1$ and $x_2$ such that $x_1 \leq x_2$ component wise, it follows from the theory of monotone dynamical systems (see \cite{Smith1995}) that $\tilde{D}e^{\tau A}x_1 \leq \tilde{D}e^{\tau A}x_2$. Thus, for any such $x_1$ and $x_2$ and any $y \leq 1$, 
\begin{align}\label{eq:discrete_inequality}
    D_{\text{eco}}(y)e^{\tau A} x_1 \leq \tilde{D}e^{\tau A} x_2 \leq  \tilde{D}e^{\tau A} x_2. 
\end{align}  
Set $\tilde{x} = x(0)$ and consider the system  $z_{k + 1} = \tilde{D}e^{\tau A}z_k$
with $z_0 = \tilde{x}$. By induction and inequality (\ref{eq:discrete_inequality}), it is clear that $0 \leq x_k \leq z_k$ for all $k \geq 1$. Furthermore, as $r(\tilde{D}e^{\tau A}) <1$, it follows from the results of \cite{patel2024spectral} that $\displaystyle \lim_{k \to \infty} z_k = 0$. Hence, $\displaystyle \lim_{k \to \infty}x_k = 0$. 

To show that $\displaystyle \lim_{t \to \infty} x(t) = 0$ for the continuous system, note that for $t \in (k\tau, (k + 1) \tau)$, $x(t) = x(k \tau)e^{A_{\text{eco}} (t - k \tau)}$ by the well-known solution to non-autonomous linear systems. As $\tau$ is a fixed value, $e^{A_{\text{eco}} (t - k \tau)}$ is bounded on the interval. Thus, as $k$ approaches infinity, it follows that $\displaystyle\lim_{t \to \infty} x(t) = 0$ as well. 
\end{proof}

\begin{theorem}  Conversely, if $r(\tilde{D}e^{\tau A}) > 1$, and at least one $s_j > 0$, then for any initial condition with $x(0)>0$ and $y_j(0)> 0$ for some $j$, we have that $\displaystyle\lim_{t\to \infty} x(t) = \infty$. 
  \end{theorem}

\begin{proof}
Suppose that $r(\tilde{D}e^{\tau A}) > 1$ and $x(0) > 0$. Similar to the proof of Result \ref{result_5}, we consider the discrete time model $x_{k+1} = D_{\text{eco}}(y(k\tau))e^{\tau A}x_k$. We claim that the sequence $x_k$ diverges to infinity. By continuity of the spectral radius, there exists $\epsilon >0$ such that $r(D_{\text{eco}}(y)e^{\tau A}) > 1$ for all $y\in B_\epsilon (1)$ where 
$B_\epsilon = \{y : |y - 1| < \epsilon \}$. As $y(k\tau)$ is a sequence created by a subset of $y$, $\displaystyle\lim_{k \to \infty} y(k\tau) = 1$ by Result \ref{result_4}. Hence, there exists a $k^* > 0$ such that $y(k\tau) \in B_\epsilon(1)$ for all $k \geq k^*$. Set $\tilde{x} = x(0)$ and consider the system $z_{k+1} = D_{\text{eco}}(y(k^*\tau))z_k$ with $z_0 = \tilde{x}$. Then for all $k \geq k^*$, $x_k \geq z_k$. By the results of \cite{patel2024spectral}, $\displaystyle\lim_{k \to \infty} z_k = \infty$, and thus it must also hold that $\displaystyle\lim_{k \to \infty} x_k = \infty$, proving the claim. Following the same argument presented in Result \ref{result_5}, this shows that $\displaystyle\lim_{t \to \infty} x(t) = \infty$.
\end{proof}

The significance of these results is that in this model, the persistence of $x$ in this eco-evolutionary model depends only on the determining persistence of the fully resistant population (when $y=1$). 
\section{Case Study of the Spotted Winged Drosophila}\label{numerical}
We apply our model to explore the evolution of resistance of an agricultural pest species, the spotted-winged Drosophila (SWD) which motivated some of this work. SWD is an invasive fruit fly of Asian origin that has become a major invasive pest of soft-skinned fruit crops such as cherries, blueberries, raspberries, and strawberries \cite{GressZalom2019}. In contrast to many other Drosophila species, which are commonly associated with damaged or overripe fruit, female SWD can deposit their eggs in fruit that is still ripening by means of a serrated ovipositor. This makes the pest particularly destructive because infestation occurs before harvest, directly reducing marketability. SWD was first detected on the North American mainland in California in 2008 and spread rapidly through major fruit-producing regions thereafter \cite{Emiljanowicz2014}.

We parameterize our model using two sets of laboratory experiments.  The first is from Emiljanowicz et al. \cite{Emiljanowicz2014} in which they measured (a) the proportion of individuals in each stage that survived to the next stage and (b) of those that transitioned, the times spent in various life stages (in Table $4$ and $5$ from Emiljanowicz).  We used these data to parameterize intrinsic mortality rates, stage transition rates, and fecundity.  Our parameterizations from this data are summarized in Table \ref{tab:continuous_params}.

\begin{table}[H]
\centering
\caption{Biological parameter values used for \textit{D. suzukii}.}
\label{tab:continuous_params}
\begin{tabular}{llll}
\hline
Parameter & Formula & Value & Description \\
\hline
$\alpha_{12}$ & $1/1.4$ & $0.714$ & egg $\to$ larva transition rate \\
$\alpha_{23}$ & $1/6.0$ & $0.167$ & larva $\to$ pupa transition rate \\
$\alpha_{34}$ & $1/5.8$ & $0.172$ & pupa $\to$ adult transition rate \\
$\lambda$ & $491.1/160$ & $3.069$ & adult fecundity parameter \\
$\mu_1$ & $0.868/1.384$ & $0.627$ & egg mortality rate \\
$\mu_2$ & $0.900/3.115$ & $0.289$ & larval mortality rate \\
$\mu_3$ & $0.944/5.836$ & $0.162$ & pupal mortality rate \\
$\mu_4$ & $0.875/10$ & $0.088$ & adult mortality rate \\
\hline
\end{tabular}
\end{table}
The second set of experiments, from Mermer et al. \cite{Mermer2021}, quantified the mortality rate of different classes of chemical compounds (insecticides) on the distinct life stages of SWD. These mortality rates were measured in populations putatively with no resistance, and can thus be understood to be $\delta_i^{11}$ for each stage. Our parameterizations from this data are found in Table \ref{tab:insecticide_survival}. 

\begin{table}[H]
\centering
\caption{Stage-specific pulse survival parameters derived from Mermer et al. \cite{Mermer2021}. Each entry is computed as $\delta_i=1-m_i$, where $m_i$ is the reported mortality proportion in stage $i$.}
\label{tab:insecticide_survival}
\begin{tabular}{lcccc}
\hline
Insecticide & $\delta_1$ (egg) & $\delta_2$ (larva) & $\delta_3$ (pupa) & $\delta_4$ (adult) \\
\hline
Malathion & $0.149$ & $0.276$ & $0.020$ & $0.100$ \\
Zeta-cypermethrin & $0.478$ & $0.405$ & $0.008$ & $0.003$ \\
Spinosad & $0.172$ & $0.241$ & $0.003$ & $0.052$ \\
Spinetoram & $0.156$ & $0.189$ & $0.008$ & $0.050$ \\
Methomyl & $0.157$ & $0.230$ & $0.000$ & $0.010$ \\
Cyantraniliprole & $0.272$ & $0.308$ & $0.050$ & $0.050$ \\
Fenpropathrin & $0.370$ & $0.335$ & $0.050$ & $0.050$ \\
Phosmet & $0.089$ & $0.085$ & $0.005$ & $0.010$ \\
Cyclaniliprole & $0.327$ & $0.282$ & $0.300$ & $0.300$ \\
\hline
\end{tabular}
\end{table}
Given the population structure of SWD, we utilize the four-dimensional version of system (\ref{eq:pulse}) to test which insecticides pose the greatest risk to evolution based on the stage in which selection occurs. Note that as very little is known about the mechanisms of resistance, we must make assumptions on the evolutionary parameters $s$ and $h$, which measures the stage-dependent selection and heterozygosity, respectively. Throughout, we set $h = 0.5$. We also make a small adjustment to the pulse survival rate for the insectide Methomyl in its effect on pupae, as $\delta_3 = 0.000$. This means that no pupae survive the pulse, and in order for the computer model to successfully illustrate results and to satisfy the assumption that ($\bar{\delta}_i>0$), we change this value to $\delta_3 = 0.0001$. 

To explore how selection at each life stage impacts the rate of resistance development for each insecticide, we consider four examples. All growth and transition parameters are fixed as in Table \ref{tab:continuous_params}, and pulse survival parameters are as in Table \ref{tab:insecticide_survival} for each numerical exploration. In each, initial population density was set to $x_i(0) = 80$ for all four stages, and all initial frequency resistances were fixed as $y_i(0) = 0.1$. Time between pulse $\tau$ was set to 20 days, and as stated above, heterozygosity coefficient $h$ was fixed at $0.5$. The selection coefficient combinations, $\{s_i\}$ are chosen for the four examples in the following way:
\begin{enumerate}
    \item[(a)] Selection occurs in eggs - $s_1 = 0.9$, $s_2, s_3, s_4 = 0$
    \item[(b)] Selection occurs in larvae - $s_2 = 0.9$, $s_1, s_3, s_4 = 0$
    \item[(c)] Selection occurs in pupae - $s_3 = 0.9$, $s_1, s_2, s_4 = 0$
    \item[(d)] Selection occurs in adult - $s_4 = 0.9, s_1, s_2, s_3 = 0$
\end{enumerate}
Results of the numerical illustrations are shown below in Figure \ref{fig:resistance_by_selection}. In each, the resistance among the adult population is displayed. The trends seen for the time to build resistance of each insecticide are the same for the life stages of eggs, larvae and pupae, in both ordering and length. They only differ by amplitude of the pulse - with life stages closure to the selected stage enjoying a sharper jump rather than a slow increase (see the difference between the pulse in the selection in adults vs. the selection in eggs). Because we are interested in what insecticide will build resistance the slowest and at what overall speed resistance will build, it will suffice to only consider the adult population. 

\begin{figure}[h]
    \centering
    \includegraphics[width=1\linewidth]{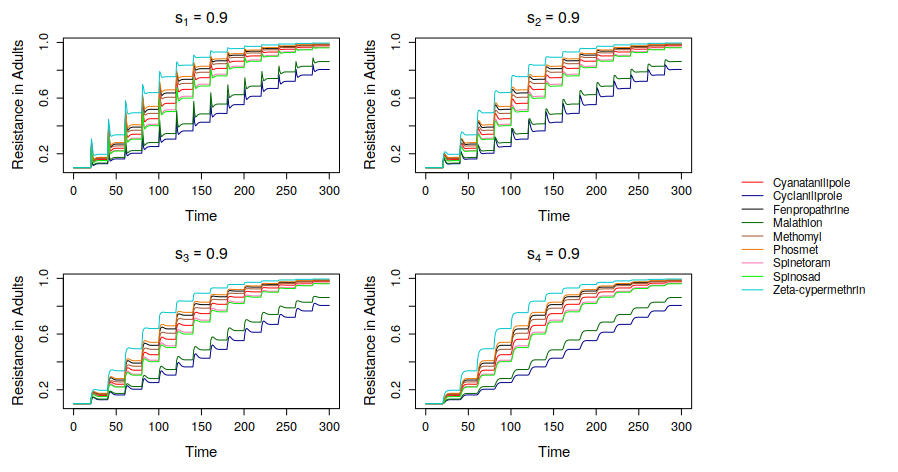}
    \caption{Plots of resistance proportion over time in adult populations varied by which stage resistance is selected for. From left to right, they are (a) selection in eggs, (b) selection in larvae, (c) selection in pupae and (d) selection in adults. Each of the nine insecticides listed in Table \ref{tab:insecticide_survival} is implemented in a separate curve.}
    \label{fig:resistance_by_selection}
\end{figure}

We note some key takeaways from the results illustrated in Figure \ref{fig:resistance_by_selection}. First, we note that our plots confirm our assertion in Result \ref{result_5}. In every example, each curve is periodically increasing toward 1. This includes the curve for Methomyl in the bottom left subfigure, where selection occurs in the pupa stage. Because Methomyl has such a low survival rate in pupae, resistance builds very slowly, as even SWD pupae with two resistant alleles only have a slightly higher than 0.01\% survival rate. However, as shown in Figure \ref{fig:methomyl_zoom}, when we plot only the Methomyl curve, we can indeed see increase in resistance frequency. 

\begin{figure}[h]
    \centering
    \includegraphics[width=0.5\linewidth]{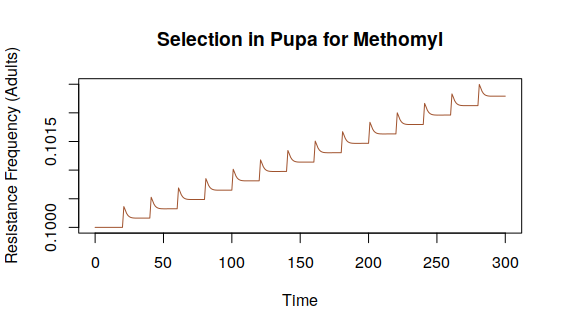}
    \caption{Resistance proportion over time in adult populations when selection occurs in pupa population and treated with Methomyl. This curve is an isolated version of that presented in Figure \ref{fig:resistance_by_selection} (c).}
    \label{fig:methomyl_zoom}
\end{figure}

We also notice that overall, resistance grows the most quickly when selection occurs in larvae. At the 300 day mark, even the slowest resistant building insecticide, Cyclaniliprole, had begun to plateau toward the 100\% frequency level. To understand why that may be the case, we note that across all insecticides, survival rate for larvae, $\delta_2$, is significantly higher than survival rates of pupae and adults. Also, although $\delta_2$ is similar in value to the survival rate in egg populations, $\delta_1$, the transition rate from egg to larva is notably faster than the transition rate from larva to pupa. Together, these observations suggest that selection in a stage with slower transition rate and higher insecticide survival leads to faster resistance building. 

Finally, we compare the effect of two insecticides of note. We see that Cyclaniliprole and Zeta-Cypermethrin have nearly opposite impacts on resistance building. In egg and larva selection, treatment with Zeta-cypermethrin leads to among the quickest rates of increased resistance. However, in pupa and adult selection, resistance builds slowly when Zeta-cypermethrin is implemented.  We notice that Zeta-cypermethrin has high survival rates in egg and larva populations, but low survival rates in pupa and adult. 

In contrast, when Cyclaniliprole is used to treat SWD, we see that resistance builds slowly when selection occurs in eggs and larvae, but quickly when it occurs in pupae and adults. Cyclaniliprole has the highest survival rates in pupae and adults, but unlike Zeta-cypermethrin, these rates are similar to those for eggs and larvae. Together, this suggests that treatment dynamics are more complex than just inspecting survival rates at the selection stage. Instead, ecological parameters seem to play an important role, and numerical exploration would be useful for determining the optimal insecticide for a given circumstance.

\section{Discussion}
In this work, we proposed an $n$-dimensional model of a stage-structured population. Along with continuous fecundity, mortality, and developmental transition, the model considers the impact of pulsed disturbances on the population from both an ecological (population density) and evolutionary (resistance frequency) standpoint. We leveraged the analysis of several sub-models which isolated specific aspects of the larger system to prove theoretical results. Namely, we established that the frequency of the resistant allele in the population will always approach 1. We also obtained conditions that guaranteed the extinction or persistence of the population, which evolutionarily only depended on the dynamics at full-resistance. Finally, we applied the model to insecticide treatment of spotted-winged Drosophila (SWD) using experimental data. Through numerical illustrations, we were able to explore the impact of a variety of insecticides given different evolutionary selection dynamics to observe the rate at which resistance built. \\


Our framework and results bring together several strands of ecological and evolutionary theory that have largely developed separately. Our results are in line with some classical population-genetic theory of the fixation of alleles under selection \cite{haldane1927selection, kimura1962probability}. More recently, Li et al \cite{li2016lifehistory} studied the fixation with structure in the population, highlighting how selection on reproduction versus survival can have different fixation rates.  In our context, we only examined survival difference amongst genotypes.  
From an ecological perspective, previous work has examined persistence of structured populations, in a community context \cite{HofbauerSchreiber2010Robust}. However, these results are not directly applicable in our context.  Their results assume an irreducibility condition that is not met when one considers explicit genetic structure. This is because, in the absence of mutation, if the population is of entirely one genotype, then it will never produce the other genotypes. In this work, persistence is simplified since we only consider one population and not the broader eco-evolutionary dynamics of a community of interacting structured populations.  
Our model connects these perspectives in a setting in which the demographic and selective effects of disturbance occur at discrete times. In the absence of the pulse disturbance, the ecological dynamics are decoupled from genotype frequencies; at each disturbance, however, stage- and genotype-specific mortality simultaneously alters population abundance and imposes selection. Consequently, the ecological and evolutionary components become coupled through the repeated disturbances, with their effects propagated by the continuous dynamics between pulses. This formulation is particularly relevant to systems in which management interventions simultaneously reduce abundance and select for resistance.

The numerical exploration presented in Section \ref{numerical} leveraged prior empirical work to parameterize our model for both ecological life cycle traits and the efficacy of different insecticides. If the model were to be applied to a different application than SWD, it would be useful for prior experiments to be done as in \cite{Emiljanowicz2014} and \cite{Mermer2021} to ensure realistic results. Furthermore, with respect to SWD, we did not have experimental data to inform our choices for heterozygosity or selection coefficients. To address this, we estimated $h = 0.5$ as a middle value, and varied where selection occurred to compare dynamics. However, to more accurately predict the growth of resistance in an SWD population, more precise values would need to be utilized. For example, since the survival proportions are known for each life-stage in a wholly resistant population, experimenters may be able to measure the increase in survival rate in a population across pulses through the LC50, or the concentration of the insecticide which will be lethal to 50\% of the population \cite{Rozman2010}. Alternatively, if genomic information is available, researchers could use techniques similar to those in \cite{Ruwende1995} to estimate the rate of insecticide survival in both homogeneous and heterogeneous members of the population at each life stage. This work would provide valuable empirical results to both heterozygosity and selection coefficients that were chosen via assumption in our examples and allow for significant improvements to predictions of rates of resistance growth over time.

In future work, fitness trade-offs associated with resistance and density-dependent population dynamics could be incorporated. In the present framework, resistant individuals may have a survival advantage during treatment without experiencing a corresponding fitness cost between pulses. However, resistance may be associated with reduced fecundity, increased mortality, or altered developmental rates. Incorporating such costs would introduce competition, with treatment favoring resistance while the dynamics between treatments may favor the susceptible population. Consequently, the frequency of the resistant allele may not necessarily approach one, and the long-term evolutionary outcome may instead depend on the frequency and strength of treatment and the magnitude of the fitness cost. Similarly, incorporating density dependence would allow population recovery between disturbances to depend on population abundance and could also reveal how ecological dynamics interacts with resistance evolution.

The evolutionary component of the model could also be extended to include mutation and a more realistic genetic basis for resistance. Our current model assumes that resistant and susceptible alleles are already present and follows their deterministic frequency dynamics. A stochastic treatment of mutation would allow resistance to arise when resistant alleles are initially absent or rare. In addition, we could have more than one locus as more empirical information becomes available. Recent work on \textit{Drosophila suzukii} has identified multiple molecular mechanisms associated with pyrethroid and spinosad resistance, including changes in genes related to detoxification and cuticular processes, and has suggested a potential role for alternative splicing \cite{Tabuloc2024}. Incorporating such mechanisms could provide a closer connection between the mathematical representation of resistance and its underlying biology.



\subsection*{Acknowledgments} SP and AA were supported by the Oregon State University College of Science Research and Innovation Seed (SciRIS) Stage 2 grant. ADL was supported by the Muhlenberg College Faculty Summer Research Grant. 

\newpage
\appendix
\section{Supplement: Details of Model Derivation}
To build on the work of Patel et al. \cite{Patel2025}, we generalize the two-stage eco-evolutionary model where pulsed treatments and selection only occurred in the adult stage. We generalize to $n$ stages, in which each stage is impacted, possibly differently, by the pulsed disturbance. More specifically, between disturbance events, the population evolves continuously according to ecological dynamics, while at discrete times it experiences a pulse disturbance through stage- and genotype-dependent survival.

\subsection{Dynamics between pulses}
We model a stage and genotype-structured population.  We assume there are $n$ total stages and, at the genetic level, we consider a single locus with two alleles. Thus, individuals in each stage belong to one of three genotypes: $11,$ $12$ and $22.$ Let $x_k^{ij}$ be the population density of individuals in stage $k$ of genotype $ij$.  Our formulation follows the general framework of stage-structured matrix population models commonly used in ecology and demography \cite{Caswell2001}.
Individuals in stage \(i\) die at a per capita death rate \(\mu_i\), and for \(i=1,\dots,n-1\), they progress from stage \(i\) to stage \(i+1\) at a per capita growth rate \(\alpha_{i,i+1}\). We assume these rates are independent of genotype. 
Adults are assumed to be in stage \(n\), and mating between two adults contributes new individuals of genotype $ij$ to stage \(1\) at a rate \(f_{ij}(x_n^{11}, x_n^{12}, x_n^{22})\), that is a function of the adult population density of the three genotypes. 

Under these assumptions, the stage and genotype-structured dynamics are given by
\begin{equation}
\begin{aligned}
\frac{dx_1^{ij}}{dt} &= f_{ij} (x_n^{11}, x_n^{12}, x_n^{22}) - (\mu_1+\alpha_{12})x_1^{ij},\\
\frac{dx_2^{ij}}{dt} &= \alpha_{12}x_1^{ij} - (\mu_2+\alpha_{23})x_2^{ij},\\
&\ \vdots\\
\frac{dx_{n-1}^{ij}}{dt} &= \alpha_{n-2,n-1}x_{n-2}^{ij}-(\mu_{n-1}+\alpha_{n-1,n})x_{n-1}^{ij},\\
\frac{dx_n^{ij}}{dt} &= \alpha_{n-1,n}x_{n-1}^{ij}-\mu_n x_n^{ij}.
\end{aligned}
\label{eq:eco_total}
\end{equation}
System \eqref{eq:eco_total} describes the transitions of individuals across developmental stages and mortality. 

Next, we define the fecundity functions $f_{ij}(x_n^{11}, x_n^{12}, x_n^{22})$.  We assume that each individual has a fecundity rate $\lambda$ that is genotype-independent and that there is random mating, as in Mendelian inheritance. Under this assumption, we arrive at

\begin{subequations}\label{fecundity}
\begin{align}
         f_{11} = &\lambda \big [ x_n^{11}(\frac{2x_n^{11} + x_n^{12}}{2x_n}) + \frac{1}{2}x_n^{12}(\frac{2x_n^{11} + x_n^{12}}{2x_n}) \big] \\
    f_{12} = &\lambda \big [ x_n^{11}(\frac{2x_n^{22} + x_n^{12}}{2x_n}) + x_n^{12}(\frac{x_n^{11} + x_n^{12} + x_n^{22}}{2x_n}) + x_n^{22}(\frac{2x_n^{11} + x_n^{12}}{2x_n}) \big ]\\
    f_{22} = & \lambda \big [ x_n^{22}(\frac{2x_n^{22} + x_n^{12}}{2x_n}) + \frac{1}{2}x_n^{12}(\frac{2x_n^{22} + x_n^{12}}{2x_n}) \big]
\end{align}
       \end{subequations}

Here, $x_n(t) = x_n^{11}(t) + x_n^{12}(t) + x_n^{22}(t)$ is the total adult population. Similarly, we let $x_k(t)=x_k^{11}(t) + x_k^{12}(t) + x_k^{22}(t)$ for $k = 1, 2...n-1$ be the total population density in stage $k$ (of all three genotypes). 

In our eco-evolutionary model, we track $x(t)=x_1(t), x_2(t), ...x_n(t)$ as the ecological variable of total population density in each stage. For the evolutionary variable, we define $y_k(t)$ as the frequency of allele 2 at time $t$ in stage $k$, which is 

\begin{equation}\label{yeqn}
    y_k=\frac{2x_k^{22} + x_k^{12}}{2x_k}
\end{equation}

We use quotient rule to derive $\frac{dy_k}{dt}$. With making some substitutions, we can express the ecological and evolutionary dynamics between the pulses as 

\begin{align}
\frac{dx}{dt}=&A_{\mathrm{eco}}x\\
\frac{dy}{dt}=&A_{\mathrm{evo}}(x)y
\end{align}
where
\begin{equation}
A_{\mathrm{eco}}=
\begin{pmatrix}
-(\mu_1+\alpha_{12}) & 0 & 0 & \cdots & 0 & \lambda\\
\alpha_{12} & -(\mu_2+\alpha_{23}) & 0 & \cdots & 0 & 0\\
0 & \alpha_{23} & -(\mu_3+\alpha_{34}) & \cdots & 0 & 0\\
\vdots & \vdots & \vdots & \ddots & \vdots & \vdots\\
0 & 0 & 0 & \cdots & -(\mu_{n-1}+\alpha_{n-1,n}) & 0\\
0 & 0 & 0 & \cdots & \alpha_{n-1,n} & -\mu_n
\end{pmatrix}.
\label{eq:eco_matrix_final}
\end{equation}
 and
\begin{equation}
A_{\mathrm{evo}}(x)=
\begin{pmatrix}
-\lambda \dfrac{x_n}{x_1} & 0 & 0 & \cdots & 0 & \lambda \dfrac{x_n}{x_1}\\[0.8em]
\alpha_{12}\dfrac{x_1}{x_2} & -\alpha_{12}\dfrac{x_1}{x_2} & 0 & \cdots & 0 & 0\\[0.8em]
0 & \alpha_{23}\dfrac{x_2}{x_3} & -\alpha_{23}\dfrac{x_2}{x_3} & \cdots & 0 & 0\\
\vdots & \vdots & \vdots & \ddots & \vdots & \vdots\\
0 & 0 & 0 & \cdots & \alpha_{n-1,n}\dfrac{x_{n-1}}{x_n} & -\alpha_{n-1,n}\dfrac{x_{n-1}}{x_n}
\end{pmatrix}.
\label{eq:evo_matrix_explicit}
\end{equation}

\subsection{Pulse Disturbance}\noindent
We now incorporate periodic disturbance into the model. We assume that the population experiences a pulse every \(\tau\) time units. At each pulse, individuals survive with probabilities that depend on both stage and genotype and the effect of the pulse occurs instantaneously in time.  This is an approximation for when the disturbance impacts are relatively quick compared to the continuous dynamics.

Let \(\delta_k^{ij}\) denote the survival probability of an individual in stage \(k\) with genotype \(ij\in\{11,12,22\}\) due to the pulse disturbance. For each stage \(k\), we assume that
\begin{equation}
\delta_k^{11}\leq \delta_k^{12}\leq \delta_k^{22},
\label{eq:delta_order}
\end{equation}
so that individuals carrying allele \(2\) have greater or equal survival during the disturbance.  That is, 1 is the resistance allele.
\\\\
Then the post-pulse genotype densities are given by
\begin{equation}
x_k^{ij}(\ell\tau^+)=\delta_k^{ij}\,x_k^{ij}(\ell\tau^-),
\qquad pq\in\{11,12,22\}.
\label{eq:geno_pulse_update}
\end{equation}
where $x_k^{ij}(\ell\tau^-)$ is a left-hand limit and $x_k^{ij}(\ell\tau^+)$ is a right hand limit, as in classical impulse differential equations.

Summing over the three genotypes, the total population in stage \(k\) after the pulse is
\begin{equation}
x_k(\ell\tau^+)
=
\delta_k^{22}x_k^{22}(\ell\tau^-)
+
\delta_k^{12}x_k^{12}(\ell\tau^-)
+
\delta_k^{22}x_k^{22}(\ell\tau^-).
\label{eq:xi_pulse_before_HW}
\end{equation}
To express the pulse update in terms of total population density and allele frequency, we assume that genotype frequencies within each stage can be approximated by Hardy--Weinberg proportions. Thus, if \(y_k\) denotes the frequency of allele \(A\) in stage \(k\), then the corresponding frequencies for genotype 22, 12, and 11 are
\[
y_k^2,\qquad 2y_k(1-y_k),\qquad (1-y_k)^2
\]
respectively. Hence,
\begin{equation}
x_k^{22}=y_k^2x_k,\qquad
x_k^{12}=2y_k(1-y_k)x_k,\qquad
x_k^{11}=(1-y_k)^2x_k.
\label{eq:HW_stage_i}
\end{equation}
Substituting \eqref{eq:HW_stage_i} into \eqref{eq:xi_pulse_before_HW} yields
\begin{equation}
x_k(k\tau^+)
=
x_k(k\tau^-)
\left[
y_i^2\delta_i^{22}
+
2y_i(1-y_i)\delta_i^{12}
+
(1-y_i)^2\delta_i^{11}
\right].
\label{eq:xi_pulse_expand}
\end{equation}
We therefore define the mean stage-specific survival by
\begin{equation}
\bar{\delta}_i
=
y_i^2\delta_i^{22}
+
2y_i(1-y_i)\delta_i^{12}
+
(1-y_i)^2\delta_i^{11},
\label{eq:delta_bar}
\end{equation}
so that the post-pulse ecological update becomes
\begin{equation}
x_k(\ell\tau^+)=\bar{\delta}_k\,x_k(\ell\tau^-).
\label{eq:xi_pulse_final}
\end{equation}
where we note that $\bar{\delta}_k$ depends on $y$.  Thus, the pulse acts on the ecological component by scaling the stage density according to the genotype-weighted average survival in that stage.

We now derive the effect of the pulse on the stage-specific allele frequencies. We recall that
\[
y_k(\ell\tau^+)
=
\frac{2x_k^{22}(\ell\tau^+)+x_k^{12}(\ell\tau^+)}{2x_k(\ell\tau^+)}.
\]
Using the genotype-level pulse update \eqref{eq:geno_pulse_update}, we obtain
\begin{equation}
y_k(\ell\tau^+)
=
\frac{2\delta_k^{22}x_k^{22}(\ell\tau^-)+\delta_k^{12}x_k^{12}(\ell\tau^-)}{2x_k(\ell\tau^+)}.
\label{eq:yi_pulse_start}
\end{equation}
Substituting the Hardy--Weinberg expressions from \eqref{eq:HW_stage_i} into \eqref{eq:yi_pulse_start} gives
\begin{equation}
y_i(k\tau^+)
=
\frac{2\delta_k^{22}y_k^2x_k+2\delta_k^{12}y_k(1-y_k)x_k}{2x_k(\ell\tau^+)},
\label{eq:yi_pulse_HW_sub}
\end{equation}
where, for simplicity, \(x_k\) and \(y_k\) are evaluated at time \(\ell\tau^-\).
\\
Using the post-pulse ecological update \eqref{eq:xi_pulse_final}, we have
\[
x_k(k\tau^+)=\bar{\delta}_k x_k(k\tau^-)=\bar{\delta}_k x_k.
\]
Hence \eqref{eq:yi_pulse_HW_sub} becomes
\begin{equation}
y_k(\ell\tau^+)
=
\frac{2\delta_k^{22}y_k^2x_k+2\delta_k^{12}y_k(1-y_k)x_k}{2\bar{\delta}_k x_k}.
\label{eq:yi_pulse_simplify1}
\end{equation}
Factoring out \(y_k\), we obtain
\begin{equation}
y_k(k\tau^+)
=
y_k
\frac{\delta_k^{22}y_k+\delta_k^{12}(1-y_k)}{\bar{\delta}_k}.
\label{eq:yi_pulse_simplify2}
\end{equation}
We now define the marginal survival associated with allele \(1\) by
\begin{equation}
\bar{\delta}_k^{\,m}
=
\delta_k^{22}y_k+\delta_k^{12}(1-y_k).
\label{eq:delta_bar_m}
\end{equation}
Therefore, the post-pulse evolutionary update can be written as
\begin{equation}
y_k(\ell\tau^+)
=
y_k(\ell\tau^-)\frac{\bar{\delta}_k^{\,m}}{\bar{\delta}_k}.
\label{eq:yi_pulse_final}
\end{equation}
Equation \eqref{eq:yi_pulse_final} shows that the pulse changes the allele frequency according to the ratio between the marginal survival of allele 2) and the mean survival in the stage. In particular, if individuals carrying allele 2 survive the pulse better on average than the stage population as a whole, then the frequency of allele 2 increases after the disturbance.

Altogether, at pulse times \(t=\ell\tau\), the ecological component is updated according to 
\begin{align}
x(\ell\tau^+)&=D_{\mathrm{eco}}(y)\,x(\ell\tau^-)\\
y(\ell\tau^+)&=D_{\mathrm{evo}}(y)\,y(\ell\tau^-),
\label{pulseeqns}
\end{align}
where
\begin{equation}
D_{\mathrm{eco}}(y)=
\operatorname{diag}\bigl(\bar{\delta}_1,\bar{\delta}_2,\dots,\bar{\delta}_n\bigr).
\label{eq:eco_pulse_diag}
\end{equation}
and
\begin{equation}
D_{\mathrm{evo}}(y)=
\operatorname{diag}\left(
\frac{\bar{\delta}_1^{\,m}}{\bar{\delta}_1},
\frac{\bar{\delta}_2^{\,m}}{\bar{\delta}_2},
\dots,
\frac{\bar{\delta}_n^{\,m}}{\bar{\delta}_n}
\right).
\label{eq:evo_pulse_diag}
\end{equation}
Therefore, the full coupled system takes the form
\begin{equation}
\begin{cases}
\dfrac{dx}{dt}=A_{\mathrm{eco}}x,\\[0.5em]
\dfrac{dy}{dt}=A_{\mathrm{evo}}(x)\,y,
\end{cases}
\qquad t\neq \ell\tau,
\label{eq:continuous_coupled}
\end{equation}
\begin{equation}
\begin{cases}
x(\ell\tau^+)=D_{\mathrm{eco}}(y)\,x(\ell\tau^-),\\[0.5em]
y(\ell\tau^+)=D_{\mathrm{evo}}(y)\,y(\ell\tau^-).
\end{cases}
\label{eq:discrete_coupled}
\end{equation}
Systems \eqref{eq:continuous_coupled} and \eqref{eq:discrete_coupled} together define the baseline eco-evolutionary model with periodic disturbance.

\newpage

\bibliographystyle{abbrv}
\bibliography{Reference}

@article{borgsteede1989lack,
  title={Lack of reversion of a benzimidazole resistant strain of {Haemonchus contortus} after six years of levamisole usage},
  author={Borgsteede, F. H. M. and Duyn, S. P. J.},
  journal={Research in Veterinary Science},
  volume={47},
  number={2},
  pages={270--272},
  year={1989},
  doi={10.1016/S0034-5288(18)31218-9}
}

@article{dilks2021newly,
  title={Newly identified parasitic nematode beta-tubulin alleles confer resistance to benzimidazoles},
  author={Dilks, Clayton M. and Koury, Emily J. and Buchanan, Claire M. and Andersen, Erik C.},
  journal={International Journal for Parasitology: Drugs and Drug Resistance},
  volume={17},
  pages={168--175},
  year={2021},
  doi={10.1016/j.ijpddr.2021.09.006}
}

@article{tang_periodic_lv_2002,
  author   = {TANG, SANYI and CHEN, LANSUN},
  title    = {THE PERIODIC PREDATOR-PREY LOTKA–VOLTERRA MODEL WITH IMPULSIVE EFFECT},
  journal  = {Journal of Mechanics in Medicine and Biology},
  volume   = {02},
  number   = {03n04},
  pages    = {267-296},
  year     = {2002},
  doi      = {10.1142/S021951940200040X},
  url      = {https://doi.org/10.1142/S021951940200040X},
  eprint   = { https://doi.org/10.1142/S021951940200040X}
}

@article{foo2014evolution,
  author  = {Foo, Jasmine and Michor, Franziska},
  title   = {Evolution of acquired resistance to anti-cancer therapy},
  journal = {Journal of Theoretical Biology},
  volume  = {355},
  pages   = {10--20},
  year    = {2014},
  doi     = {10.1016/j.jtbi.2014.02.025}
}

@book{WHO2006preventive,
  author    = {{World Health Organization}},
  title     = {Preventive Chemotherapy in Human Helminthiasis: Coordinated Use of Anthelminthic Drugs in Control Interventions: A Manual for Health Professionals and Programme Managers},
  publisher = {World Health Organization},
  address   = {Geneva},
  year      = {2006},
  isbn      = {9789241547109}
}

@article{patel2025anthelmintic,
  author  = {Patel, Swati and Lyberger, Kelsey and Vegvari, Carolin and Gulbudak, Hayriye},
  title   = {Eco-evolutionary dynamics of anthelmintic resistance in soil-transmitted helminths},
  journal = {Theoretical Population Biology},
  volume  = {163},
  pages   = {80--90},
  year    = {2025},
  doi     = {10.1016/j.tpb.2025.03.006}
}

@article{lakmeche_chemo_periodic_2000,
  author    = {Lakmeche, A and Arino, O},
  address   = {WATERLOO},
  issn      = {1201-3390},
  journal   = {Dynamics of continuous, discrete and impulsive systems},
  language  = {eng},
  number    = {2},
  pages     = {265-287},
  publisher = {Watam Press},
  title     = {Bifurcation of non trivial periodic solutions of impulsive differential equations arising chemotherapeutic treatment},
  volume    = {7},
  year      = {2000}
}

@article{patel_spectral_2024,
  title    = {On the spectral radius properties of a key matrix in periodic impulse control},
  volume   = {187},
  issn     = {0167-6911},
  url      = {https://www.sciencedirect.com/science/article/pii/S0167691124000690},
  doi      = {10.1016/j.sysconle.2024.105781},
  urldate  = {2025-01-18},
  journal  = {Systems \& Control Letters},
  author   = {Patel, Swati and De Leenheer, Patrick},
  month    = may,
  year     = {2024},
  vol      = {187},
  issn     = {105781}
}

@article{mermer2021timing,
 title={Timing and order of different insecticide classes drive control of Drosophila suzukii; a modeling approach},
 author={Mermer, Serhan and Pfab, Ferdinand and Tait, Gabriella and Isaacs, Rufus and Fanning, Philip D and Van Timmeren, Steven and Loeb, Gregory M and Hesler, Stephen P and Sial, Ashfaq A and Hunter, Jamal H and others},
 journal={Journal of Pest Science},
 volume={94},
 number={3},
 pages={743--755},
 year={2021},
 publisher={Springer}
}

@article{mailleret_pulsed_biocontrol_2009,
  title    = {Global stability and optimisation of a general impulsive biological control model},
  journal  = {Mathematical Biosciences},
  volume   = {221},
  number   = {2},
  pages    = {91-100},
  year     = {2009},
  issn     = {0025-5564},
  doi      = {https://doi.org/10.1016/j.mbs.2009.07.002},
  url      = {https://www.sciencedirect.com/science/article/pii/S0025556409001151},
  author   = {Ludovic Mailleret and Frédéric Grognard}
}

@article{ren2017tumour,
  author  = {Ren, H.-P. and Yang, Y. and Baptista, M. S. and Grebogi, C.},
  title   = {Tumour chemotherapy strategy based on impulse control theory},
  journal = {Philosophical Transactions of the Royal Society A: Mathematical, Physical and Engineering Sciences},
  volume  = {375},
  number  = {2088},
  pages   = {20160221},
  year    = {2017},
  doi     = {10.1098/rsta.2016.0221}
}

@article{gao2017mass,
  author  = {Gao, Daozhou and Lietman, Thomas M. and Dong, Cheng-Ping and Porco, Travis C.},
  title   = {Mass drug administration: the importance of synchrony},
  journal = {Mathematical Medicine and Biology},
  volume  = {34},
  number  = {2},
  pages   = {241--260},
  year    = {2017},
  doi     = {10.1093/imammb/dqw005}
}

@article{li2016lifehistory,
  author  = {Li, Xiang-Yi and Kurokawa, Shun and Giaimo, Stefano and Traulsen, Arne},
  title   = {How Life History Can Sway the Fixation Probability of Mutants},
  journal = {Genetics},
  volume  = {203},
  number  = {3},
  pages   = {1297--1313},
  year    = {2016},
  doi     = {10.1534/genetics.116.188409}
}

@article{haldane1927selection,
  author  = {Haldane, J. B. S.},
  title   = {A Mathematical Theory of Natural and Artificial Selection, Part V: Selection and Mutation},
  journal = {Proceedings of the Cambridge Philosophical Society},
  volume  = {23},
  number  = {7},
  pages   = {838--844},
  year    = {1927},
  doi     = {10.1017/S0305004100015644}
}

@article{kimura1962probability,
  title={On the probability of fixation of mutant genes in a population},
  author={Kimura, Motoo},
  journal={Genetics},
  volume={47},
  number={6},
  pages={713--719},
  year={1962},
  doi={10.1093/genetics/47.6.713}
}

@article{HofbauerSchreiber2010Robust,
  author  = {Hofbauer, Josef and Schreiber, Sebastian J.},
  title   = {Robust Permanence for Interacting Structured Populations},
  journal = {Journal of Differential Equations},
  year    = {2010},
  volume  = {248},
  number  = {8},
  pages   = {1955--1971},
  doi     = {10.1016/j.jde.2009.11.010}
}

@article{pastore2021evolution,
  author  = {Pastore, Abigail I. and Barab{\'a}s, Gy{\"o}rgy and
             Bimler, Malyon D. and Mayfield, Margaret M. and
             Miller, Thomas E.},
  title   = {The Evolution of Niche Overlap and Competitive Differences},
  journal = {Nature Ecology \& Evolution},
  year    = {2021},
  volume  = {5},
  number  = {3},
  pages   = {330--337},
  doi     = {10.1038/s41559-020-01383-y}
}

@article{cortez2016magnitude,
  author  = {Cortez, Michael H.},
  title   = {How the Magnitude of Prey Genetic Variation Alters
             Predator--Prey Eco-Evolutionary Dynamics},
  journal = {The American Naturalist},
  year    = {2016},
  volume  = {188},
  number  = {3},
  pages   = {329--341},
  doi     = {10.1086/687393}
}

@article{patel2018partitioning,
  author  = {Patel, Swati and Cortez, Michael H. and Schreiber, Sebastian J.},
  title   = {Partitioning the Effects of Eco-Evolutionary Feedbacks on Community Stability},
  journal = {The American Naturalist},
  year    = {2018},
  volume  = {191},
  number  = {3},
  pages   = {381--394},
  doi     = {10.1086/695834}
}

@article{patel2026persistence,
  author  = {Patel, Swati and Govaert, Lynn and Lyberger, Kelsey and
             Luque, Victor J. and Duthie, A. Bradley and Lion, S{\'e}bastien},
  title   = {Persistence and Near Persistence via Trait Evolution:
             Pathways to Coexistence},
  journal = {Journal of Theoretical Biology},
  year    = {2026},
  volume  = {627},
  pages   = {112443},
  doi     = {10.1016/j.jtbi.2026.112443}
}

@article{patel2018robust,
  author  = {Patel, Swati and Schreiber, Sebastian J.},
  title   = {Robust Permanence for Ecological Equations with Internal and External Feedbacks},
  journal = {Journal of Mathematical Biology},
  year    = {2018},
  volume  = {77},
  number  = {1},
  pages   = {79--105},
  doi     = {10.1007/s00285-017-1187-5}
}

@article{schreiber2018evolution,
  author  = {Schreiber, Sebastian J. and Patel, Swati and terHorst, Casey P.},
  title   = {Evolution as a Coexistence Mechanism: Does Genetic Architecture Matter?},
  journal = {The American Naturalist},
  year    = {2018},
  volume  = {191},
  number  = {3},
  pages   = {407--420},
  doi     = {10.1086/695832}
}

@article{patel2019ecoevolutionary,
  author  = {Patel, Swati and B{\"u}rger, Reinhard},
  title   = {Eco-Evolutionary Feedbacks between Prey Densities and
             Linkage Disequilibrium in the Predator Maintain Diversity},
  journal = {Evolution},
  year    = {2019},
  volume  = {73},
  number  = {8},
  pages   = {1533--1548},
  doi     = {10.1111/evo.13785}
}

@article{patel2015evolutionarily,
  author  = {Patel, Swati and Schreiber, Sebastian J.},
  title   = {Evolutionarily Driven Shifts in Communities with Intraguild Predation},
  journal = {The American Naturalist},
  year    = {2015},
  volume  = {186},
  number  = {5},
  pages   = {E98--E110},
  doi     = {10.1086/683170}
}

@article{schreiber2011community,
  author  = {Schreiber, Sebastian J. and B{\"u}rger, Reinhard and Bolnick, Daniel I.},
  title   = {The Community Effects of Phenotypic and Genetic Variation within a Predator Population},
  journal = {Ecology},
  year    = {2011},
  volume  = {92},
  number  = {8},
  pages   = {1582--1593},
  doi     = {10.1890/10-2071.1}
}

@article{lande1976natural,
  author  = {Lande, Russell},
  title   = {Natural Selection and Random Genetic Drift in Phenotypic Evolution},
  journal = {Evolution},
  year    = {1976},
  volume  = {30},
  number  = {2},
  pages   = {314--334},
  doi     = {10.1111/j.1558-5646.1976.tb00911.x}
}

@incollection{metz1996adaptive,
  title     = {Adaptive Dynamics: A Geometrical Study of the Consequences
               of Nearly Faithful Reproduction},
  author    = {Metz, Johan A. J. and Geritz, Stefan A. H. and
               Mesz{\'e}na, G{\'e}za and Jacobs, Frans J. A. and
               van Heerwaarden, J. S.},
  booktitle = {Stochastic and Spatial Structures of Dynamical Systems},
  editor    = {van Strien, S. J. and Verduyn Lunel, S. M.},
  pages     = {183--231},
  year      = {1996},
  publisher = {North-Holland},
  address   = {Amsterdam}
}

@article{dieckmann1996dynamical,
  title   = {The Dynamical Theory of Coevolution: A Derivation from
             Stochastic Ecological Processes},
  author  = {Dieckmann, Ulf and Law, Richard},
  journal = {Journal of Mathematical Biology},
  volume  = {34},
  number  = {5--6},
  pages   = {579--612},
  year    = {1996},
  doi     = {10.1007/BF02409751}
}

@article{geritz1998evolutionarily,
  title   = {Evolutionarily Singular Strategies and the Adaptive Growth
             and Branching of the Evolutionary Tree},
  author  = {Geritz, Stefan A. H. and Kisdi, {\'E}va and
             Mesz{\'e}na, G{\'e}za and Metz, Johan A. J.},
  journal = {Evolutionary Ecology},
  volume  = {12},
  number  = {1},
  pages   = {35--57},
  year    = {1998},
  doi     = {10.1023/A:1006554906681}
}

@article{patel2024spectral,
  title={On the spectral radius properties of a key matrix in periodic impulse control},
  author={Patel, Swati and De Leenheer, Patrick},
  journal={Systems \& Control Letters},
  volume={187},
  pages={105781},
  year={2024},
  publisher={Elsevier}
}

@article{Emiljanowicz2014,
  author  = {Emiljanowicz, Lisa M. and Ryan, Geraldine D. and Langille, Aaron and Newman, Jonathan},
  title   = {Development, Reproductive Output and Population Growth of the Fruit Fly Pest {Drosophila suzukii} ({Diptera}: {Drosophilidae}) on Artificial Diet},
  journal = {Journal of Economic Entomology},
  year    = {2014},
  volume  = {107},
  number  = {4},
  pages   = {1392--1398},
  doi     = {10.1603/EC13504}
}

@article{Mermer2021,
  author  = {Mermer, Serhan and Pfab, Ferdinand and Tait, Gabriella and Isaacs, Rufus and Fanning, Philip D. and Van Timmeren, Steven and Loeb, Gregory M. and Hesler, Stephen P. and Sial, Ashfaq A. and Hunter, Jamal H. and Bal, Harit Kaur and Drummond, Francis and Ballman, Elissa and Collins, Judith and Xue, Lan and Jiang, Duo and Walton, Vaughn M.},
  title   = {Timing and Order of Different Insecticide Classes Drive Control of {Drosophila suzukii}; a Modeling Approach},
  journal = {Journal of Pest Science},
  year    = {2021},
  doi     = {10.1007/s10340-020-01292-w}
}

@article{Markus1956,
    author = {Markus, L.},
    title = {Asymptotically autonomous differential systems},
    journal = {Contributions to the Theory of Nonlinear Oscillations},
    year = {1956}, 
    volume = {III}, 
    publisher = {Princeton: Princeton University Press}, 
    doi = {10.1515/9781400882175-003}
}

@book{Smith1995,
  author={Smith, H.L.},
  title = {Monotone Dynamical
Systems: An Introduction to the Theory
of Competitive and Cooperative
Systems},
  year={1995},
  publisher={Springer}
}

@book{Caswell2001,
  author    = {Caswell, Hal},
  title     = {Matrix Population Models: Construction, Analysis, and Interpretation},
  edition   = {2},
  publisher = {Sinauer Associates},
  address   = {Sunderland, Massachusetts},
  year      = {2001},
  isbn      = {0878930965}
}

@article{Patel2025,
  author  = {Patel, Swati and Lyberger, Kelsey and Vegvari, Carolin and Gulbudak, Hayriye},
  title   = {Eco-evolutionary Dynamics of Anthelmintic Resistance in Soil-transmitted Helminths},
  journal = {Theoretical Population Biology},
  year    = {2025},
  volume  = {163},
  pages   = {80--90},
  doi     = {10.1016/j.tpb.2025.03.006}
}

@article{GressZalom2019,
  author  = {Gress, Brian E. and Zalom, Frank G.},
  title   = {Identification and Risk Assessment of Spinosad Resistance in a California Population of {Drosophila suzukii}},
  journal = {Pest Management Science},
  year    = {2019},
  volume  = {75},
  number  = {5},
  pages   = {1270--1276},
  doi     = {10.1002/ps.5240}
}

@article{Taylor1990,
    author = {Taylor, Peter D.},
    title = {Allele-Frequency Change in a Class Structured Population},
    journal = {The American Naturalist},
    year = {1990},
    volume = {135},
    number = {1},
    pages = {95 -- 106}
}

@article{Barfield2011,
author = {Barfield, Michael and Holt, Robert D. and Gomulkiewicz, Richard},
title = {Evolution in Stage-Structured Populations},
journal = {The American Naturalist}, 
year = {2011}, 
volume = {177}, 
number = {4}, 
pages = {397 -- 409}, 
doi = {10.1086/658903}
}

@article{deVries2019,
    author = {de Vries, Charlotte and Caswell, Hal},
    title = {Stage-Structured Evolutionary Demography: Linking Life Histories, Population Genetics, and Ecological Dynamics},
    journal = {The American Naturalist},
    year = {2019},
    volume = {193}, 
    number = {4}, 
    pages = {545 -- 559}, 
    doi = {10.1086/701857}
}

@article{Dong2006,
    author={Dong, L. and Chen, L. and Sun, L.},
    title={{Extinction and permanence of the predator–prey system with
stocking of prey and harvesting of predator impulsively}},
    year={2006},
    journal={Mathematical Methods in the Applied Sciencs},
    volume = {29},
    pages = {415--425}
}

@article{Jiao2008,
author = {Jiao, Jj. and Chen, Ls. and Nieto, J.J.}, 
title = {Permanence and global attractivity of stage-structured predator-prey model with continuous harvesting on predator and impulsive stocking on prey}, 
journal = {Appl. Math. Mech.-Engl. Ed.}, 
volume = {29}, 
pages = {653--663},
year = {2008}, 
doi = {doi.org/10.1007/s10483-008-0509-x}
}

@article{Jiao2007,
    author={Jiao, J. and Meng, X. and Chen, L. },
    title={{A stage-structured Holling mass defence predator-prey model with impulsive perturbations on predators}},
    year={2007},
    journal={Applied Math and Computation},
    volume = {189},
    pages = {1448--1458}
}

@article{He2015,
title = {Dynamics analysis of a two-species competitive model with state-dependent impulsive effects},
journal = {Journal of the Franklin Institute},
volume = {352},
number = {5},
pages = {2090-2112},
year = {2015},
issn = {0016-0032},
doi = {doi.org/10.1016/j.jfranklin.2015.02.021},
author = {Zhi-Long He and Lin-Fei Nie and Zhi-Dong Teng}
}

@article{Jin2005,
title = {The persistence in a {L}otka–{V}olterra competition systems with impulsive},
journal = {Chaos, Solitons \& Fractals},
volume = {24},
number = {4},
pages = {1105-1117},
year = {2005},
issn = {0960-0779},
doi = {doi.org/10.1016/j.chaos.2004.09.065},
author = {Zhen Jin and Han Maoan and Li Guihua}
}

@article{Kalra2022,
author = {Kalra, Preety and Kaur, Maninderjit},
title = {Stability analysis of an eco-epidemiological SIN model with impulsive control strategy for integrated pest management considering stage-structure in predator},
journal = {International Journal of Mathematical Modelling and Numerical Optimisation},
volume = {12},
number = {1},
pages = {43-68},
year = {2022},
doi = {10.1504/IJMMNO.2022.119779}
}

@article{Li2022,
author = {Li, Zuxiong and Yang, Xue and Fu, Shengnan},
title = {Dynamical Behavior of a Predator-Prey System Incorporating a Prey Refuge with Impulse Effect},
journal = {Complexity},
volume = {2022},
number = {1},
doi = {doi.org/10.1155/2022/2422923},
year = {2022}
}

@article{dovidio2026,
    author = {D'Ovidio Long, A. and Deng, B. and Du, H.},
    title = {Dynamical Analysis of an Impulsive Model of
Cancer Cell Populations Under Radiotherapy},
    journal = {Mathematical Medicine and Biology: A Journal of the IMA},
    year = {2026},
    doi = {10.1093/imammb/dqag004}
}

@book{Rozman2010,
    author = {Rozman, Karl K. and Doull, John and Hayes Jr., Wayland J.},
    title = {Handbook of Pesticide Toxicology (Third Edition)},
    publisher = {Academic Press},
    year = {2010},
    pages = {3--101},
    doi = {10.1016/B978-0-12-374367-1.00001-X}
}

@article{Ruwende1995,
    author = {Ruwende, C. and Khoo, S.C. and Snow, R.W. and Yates, S.N.R and
    Kwiatkowski, D. and Gupta, S. and Warn, P. and Allsopp, C.E.M. and Gilbert, S.C. and Peschu, N. and Newbold, C.I. and Greenwood, B.M. and Marsh, K. and Hill, A.V.S.},
    title = {Natural selection of hemi- and heterozygotes for G6PD deficiency in Africa by resistance to severe malaria}, 
    journal = {Nature}, 
    year = {1995}, 
    pages = {246 -- 249}, 
    volume = {376}, 
    doi = {10.1038/376246a0}
}

@article{Wiman2016,
  author  = {Wiman, Nik G. and Dalton, Daniel T. and Anfora, Gianfranco
             and Biondi, Antonio and Chiu, Joanna C. and Daane, Kent M.
             and Gerdeman, Beverly and Gottardello, Angela and Hamby,
             Kelly A. and Isaacs, Rufus and Grassi, Alberto and Ioriatti,
             Claudio and Lee, Jana C. and Miller, Betsy and Stacconi,
             Marco Valerio Rossi and Shearer, Peter W. and Tanigoshi,
             Lynell and Wang, Xingeng and Walton, Vaughn M.},
  title   = {Drosophila suzukii population response to environment
             and management strategies},
  journal = {Journal of Pest Science},
  year    = {2016},
  volume  = {89},
  pages   = {653--665},
  doi     = {10.1007/s10340-016-0757-4}
}

@article{Ganjisaffar2022,
  author  = {Ganjisaffar, Fatemeh and Demkovich, Mark R. and
             Chiu, Joanna C. and Zalom, Frank G.},
  title   = {Characterization of Field-Derived Drosophila suzukii
             (Diptera: Drosophilidae) Resistance to Pyrethroids in
             California Berry Production},
  journal = {Journal of Economic Entomology},
  year    = {2022},
  volume  = {115},
  number  = {5},
  pages   = {1676--1684},
  doi     = {10.1093/jee/toac118}
}

@article{Tabuloc2024,
  author  = {Tabuloc, Christine A. and Carlson, Curtis R. and
             Ganjisaffar, Fatemeh and Truong, Cindy C. and
             Chen, Ching-Hsuan and Lewald, Kyle M. and
             Hidalgo, Sergio and Nicola, Nicole L. and
             Jones, Cera E. and Sial, Ashfaq A. and
             Zalom, Frank G. and Chiu, Joanna C.},
  title   = {Transcriptome analysis of Drosophila suzukii reveals
             molecular mechanisms conferring pyrethroid and spinosad resistance},
  journal = {Scientific Reports},
  volume  = {14},
  pages   = {19867},
  year    = {2024},
  doi     = {10.1038/s41598-024-70037-x}
}

@article{Horst1994, 
author = {Thieme, Horst R.},
title = {Asymptotically Autonomous Differential Equations in the Plane},
journal = {The Rocky Mountain Journal of Mathematics},
volume = {24},
pages = {351 -- 380},
year = {1994}, 
doi = {10.1216/rmjm/1181072470}
}

\end{document}